\documentclass{fundam}
\pdfoutput=1 

\usepackage[T5,T1]{fontenc} 
\usepackage{cleveref}
\usepackage{bussproofs}
\usepackage{csquotes}
\usepackage{mdframed}
\usepackage{mathpartir}
\usepackage{todonotes}
\usepackage{amsmath,amssymb}
\usepackage{stmaryrd}
\usepackage{mathtools}

\usepackage{url} 
\usepackage[ruled,lined]{algorithm2e}
\usepackage{graphicx}
\usepackage{tikz}
\usetikzlibrary{automata, positioning, arrows}

\newcommand{\bbN}{\mathbb{N}}

\newcommand{\atoms}{\mathbb{A}}
\newcommand{\Image}{\mathrm{Im}}

\newcommand{\cC}{\mathcal{C}}

\newcommand{\cI}{\mathcal{I}}

\newcommand{\pebble}{\mathsf{pebble}}
\newcommand{\blind}{\mathsf{blind}}

\newcommand{\Pebble}{\mathsf{Pebble}}
\newcommand{\Blind}{\mathsf{Blind}}
\newcommand{\MTT}{\mathsf{MTT}}
\newcommand{\SO}{\mathrm{SO}}
\newcommand{\yield}{{\mathtt{yield}}}

\newcommand{\relat}{\mathrm{R}}
\newcommand{\worig}{\genfrac{}{}{0pt}{}}

\theoremstyle{claimstyle}

\usepackage{todonotes}

\newcommand\withul{\lightning}

\newcommand\innsq{\mathtt{innsq}}

\begin{document}

%
\setcounter{page}{1}
\publyear{2023}
\papernumber{???}
\volume{{\color{red}\textbf{submitted preprint}}}
\issue{???}
%

\title{Refutations of pebble minimization via output languages}


\author{Sandra Kiefer\\ University of Oxford, United Kingdom\\
  sandra.kiefer{@}cs.ox.ac.uk
  \and {\fontencoding{T5}\selectfont Lê Thành Dũng (Tito) Nguyễn}\\
  École normale supérieure de Lyon, France\\
  nltd{@}nguyentito.eu
  \and
  Cécilia Pradic\\ Swansea University, United Kingdom\\ c.pradic{@}swansea.ac.uk}

\maketitle

\runninghead{S.~Kiefer, L.~T.~D.~{\fontencoding{T5}\selectfont{}Nguyễn} and C.~Pradic}{Refuting pebble minimization via output languages}

\begin{abstract}
  Polyregular functions are the class of string-to-string functions definable by
  pebble transducers, an extension of finite-state automata with outputs and
  multiple two-way reading heads (pebbles) with a stack discipline. If a polyregular function can be computed with $k$ pebbles, then its output length is bounded by a polynomial of degree $k$ in the input
  length. But
  Bojańczyk has shown that the converse fails.

  In this paper, we provide two alternative easier proofs. The first establishes
  by elementary means that some quadratic polyregular function requires 3
  pebbles. The second proof -- just as short, albeit less elementary -- shows a
  stronger statement: for every~$k$, there exists some polyregular function with
  quadratic growth whose output language differs from that of any $k$-fold
  composition of macro tree transducers (and which therefore cannot be computed
  by a $k$-pebble transducer).
  Along the way, we also refute a conjectured logical characterization of polyblind functions.
\end{abstract}

\begin{keywords}
  polyregular functions, pebble transducers, macro tree transducers
\end{keywords}

\paragraph*{Acknowledgments}

We thank Mikołaj Bojańczyk for giving us much of the initial
inspiration that led to this paper, as well as Gaëtan Douéneau-Tabot and Nathan
Lhote for stimulating discussions.

{\footnotesize L.~T.~D.~{\fontencoding{T5}\selectfont{}Nguyễn} was supported by the LABEX MILYON (ANR-10-LABX-0070) of Université de Lyon, within the program \enquote{Investissements d'Avenir} (ANR-11-IDEX-0007) operated by the French National Research Agency (ANR).}

  \section{Introduction}

Many works about \emph{transducers} -- automata that can produce output string/trees, not just recognize languages -- are concerned with:
\begin{itemize}
  \item either well-known classes of \emph{linearly growing} functions, i.e.\ $|f(x)| = O(|x|)$ -- in the string-to-string case, those are the sequential, rational and regular functions, see e.g.~\cite{MuschollPuppis};
  \item or devices with possibly (hyper)exponential growth rates: HDT0L systems~\cite{FerteMarinSenizergues,CopyfulSST,Marble}, macro tree transducers (MTTs)~\cite{Macro}, compositions of MTTs (see below)…
\end{itemize}
A middle ground is occupied by the \emph{polyregular functions}, surveyed
in~\cite{PolyregSurvey}, which derive their name from their \emph{polynomial}
growth rate. They are the functions computed by string-to-string \emph{pebble
  transducers}. Although pebble transducers have been around for two decades
(starting with the tree-to-tree version in~\cite{Pebble}, see
also~\cite{EngelfrietPebbleMacro,PebbleString,PebbleComposition}), it is only in
the late 2010s that several alternative characterizations of polyregular
functions were introduced~\cite{polyregular,msoInterpretations}. This had led to
renewed interest in this robust function class (e.g.~\cite{Zpolyreg,JordonPhD,Reduplication}).
\begin{example}
  \label{exa:innsq}
  The \emph{inner squaring} function (from~{\cite[\S6.2]{PolyregSurvey}})
  defined as
  \[ \begin{array}{llcl}
    \innsq\colon& \{a,b,\#\}^* &\to& \{a,b,\#\}^*\\
    & w_0 \# \dots \# w_n &\mapsto& (w_0)^n \# \dots \# (w_n)^n \qquad (w_0,\dots,w_n \in \{a,b\}^*)
     \end{array} \]
   -- for instance, $\innsq(aba\#baa\#bb) = abaaba\#baabaa\#bbbb$ -- can be computed by the 3-pebble\footnote{\label{ftn:off-by-one}We
     follow the convention of newer papers for two-way string transducers are
     1-pebble transducers, while some older papers would have called them 0-pebble
     transducers; thus, $k\geq1$ here.} transducer that we describe in
   \Cref{sec:innsq}. Essentially, it uses a stack of at most 3 pointers to the
   inputs, called pebbles, in order to simulate 3 nested for-loops. This means
   that in this example, the pebbles only move in a left-to-right direction; but in
   general, they can choose to move to the left or to the right depending on the
   current state, unlike the indices in a for-loop.
\end{example}

A straightforward property of $k$-pebble transducers is that their growth rate
is\footnote{This applies to strings; for trees, $(\text{output \emph{height}}) =
  O((\text{input \emph{size}})^k)$~\cite[Lemma~7]{EngelfrietPebbleMacro}.}
$O(n^k)$, and this is the best possible bound. Hence the question of
\emph{pebble minimization}: if a polyregular function (that is, defined by some
$\ell$-pebble transducer) has growth $O(n^k)$, is it always computable by some
$k$-pebble transducer? Definitely not, as Bojańczyk recently
showed~\cite[Section~3]{PolyregGrowth}: no number of pebbles suffices to compute
all polyregular functions with \emph{quadratic} growth (but those of
\emph{linear} growth require a single pebble). Moreover, he uses the tools
from~\cite{PolyregGrowth} to show
in~\cite[Theorem~6.3]{PolyregSurvey}\footnote{The paper~\cite{PolyregGrowth}
  proposes a slightly different quadratic example that requires 3 pebbles,
  called \enquote{block squaring}.} that the inner squaring function requires 3
pebbles even though $|\innsq(w)| = O(|w|^2)$.

\paragraph{Contributions}

Bojańczyk's proof of the aforementioned results is long and technical. We
propose shorter and easier arguments:
\begin{itemize}
\item We reprove in \Cref{sec:innsq-non-min} that inner squaring requires 3 pebbles.
  Our proof depends only on a few old and familiar properties of \emph{regular
    functions} -- those computed by two-way transducers, i.e.\ \mbox{1-pebble}
  transducers -- such as their closure under composition~\cite{ChytilJ77} and an
  elementary pumping lemma~\cite{Rozoy} for their \emph{output languages} (their
  sets of possible output strings).
\item In \Cref{sec:engelfrieteries}, we use first-order interpretations
  (recalled in \Cref{sec:fo-interpretations}) to construct a certain sequence
  $(f_k)_{k\geq1}$ of quadratic polyregular functions. Then, we show that $f_k$
  requires $k+1$ pebbles, as a consequence of a stronger result: the output
  language of $f_k$ differs from that of any $k$-fold composition of macro tree
  transducers. (This composition hierarchy is quite canonical and well-studied,
  see \S\ref{sec:comp-mtt}.) Again, our proof is quite short,\footnote{And
    mostly unoriginal: it consists of little adjustments to an argument by
    Engelfriet \& Maneth~\cite[\S4]{PebbleString}. In a followup
    paper~\cite{PebbleComposition}, Engelfriet mentions and corrects a mistake
    in~\cite[\S3]{PebbleString}, but it does not affect the section which is
    relevant for our purposes.} and even arguably easier to check than our
  ad-hoc argument for $\innsq$, though it is less elementary since its
  \enquote{trusted base} is larger: we use a powerful \enquote{bridge theorem}
  on MTTs from~\cite{OutputMacro}.
\end{itemize}

\paragraph{On subclasses of polyregular functions}
Some restrictions on pebble transducers ensure that:
\begin{itemize}
  \item the computed function sits at a low level of the aforementioned composition hierarchy;
  \item a pebble minimization property holds: functions of growth $O(n^k)$
    require only $k$ pebbles.
\end{itemize}
This is for instance the case for \emph{$k$-marble transducers} -- they can
compute precisely the same string-to-string functions as MTTs\footnote{Using the
  fact that string-to-string MTTs are syntactically isomorphic, up to
  insignificant details, to the copyful streaming string transducers
  of~\cite{CopyfulSST}.} with growth $O(n^k)$~\cite[\S5]{Marble} -- or for
\emph{blind pebble transducers}, which define \emph{polyblind\footnote{A name
    given by
    Douéneau-Tabot~\cite{UnaryOutput,doueneautabot2022hiding,LastPebble} to what
    {\fontencoding{T5}\selectfont Nguyễn}, Noûs and Pradic~\cite{NguyenNP21}
    originally called the \enquote{comparison-free} subclass of polyregular
    functions.} functions} (cf.~Theorem~\ref{thm:blind-min}, taken
from~\cite{NguyenNP21}). Furthermore, Douéneau-Tabot has recently proved
pebble minimization for \enquote{last pebble} transducers~\cite{LastPebble}, a
model that subsumes both marble and blind pebble transducers; and a function
computed by any such device can be obtained by a composition of \emph{two}
MTTs\footnote{See also \cite{lvl3} for a characterization of the
  string-to-string functions computed by compositions of two MTTs.} (a
consequence of~\cite[Theorem~53~(in~\S15)]{InvisiblePebbles}).

Therefore, counterexamples to pebble minimization cannot be computed by
\enquote{blind pebble} or \enquote{last pebble} transducers. Based on this
observation, and using Example~\ref{exa:innsq}, we refute the conjecture
in~\cite{NguyenNP21} about a logical characterization of polyblind functions
(Theorem~\ref{thm:polyblind-refutation}).

\paragraph{Notations}

The set of natural numbers is $\bbN=\{0,1,\dots\}$. Alphabets are always finite sets.

Let $\Sigma$ be an alphabet. We write $\varepsilon$ for the empty string and $\Sigma^*$ for the set of strings (or words) with letters in $\Sigma$, i.e.\ the free monoid over $\Sigma$; the Kleene star $(-)^*$ will also be applied to languages $L \subseteq \Sigma^*$ as part of usual regular expression syntax. 

Let $w = w_1 \dots w_n \in \Sigma^*$ ($w_i \in \Sigma$ for $i\in\{1,\dots,n\}$);
we also write $w[i] = w_i$. The length $|w|$ of $w$ is $n$, and $|w|_c$ refers
to the number of occurrences of $c\in\Sigma$ in $w$.

The \emph{output language} of a function $f \colon X \to \Sigma^*$ is $f(X) \subseteq \Sigma^*$, also denoted by $\Image(f)$.

\section{An example of string-to-string pebble transducer}
\label{sec:innsq}

We describe informally a pebble transducer that computes the inner squaring
function of Example~\ref{exa:innsq}.

The machine has a finite-state control, and starts with a single pebble
(pointer to the input) $\downarrow$ initialized to the first position in the
input. The first thing it does is to \emph{push} a second pebble $\Downarrow$
on its \emph{stack of pebbles}, resulting in the following configuration:
\begin{center}
  \begin{tikzpicture}
    \draw[thick] (0,0.4) -- (0,1) -- (12,1) -- (12,0.4) -- (0,0.4);
    \draw[thick] (1,0.4) -- (1,1);
    \draw[thick] (2,0.4) -- (2,1);
    \draw[thick] (3,0.4) -- (3,1);
    \draw[thick] (4,0.4) -- (4,1);
   \draw[thick] (5,0.4) -- (5,1);
    \draw[thick] (6,0.4) -- (6,1);
    \draw[thick] (7,0.4) -- (7,1);
    \draw[thick] (8,0.4) -- (8,1);
    \draw[thick] (9,0.4) -- (9,1);
    \draw[thick] (10,0.4) -- (10,1);
    \draw[thick] (11,0.4) -- (11,1);
    \node at (0.5,0.7) {\large $\triangleright$};
    \node at (1.5,0.7) {\large $a$};
    \node at (2.5,0.7) {\large $b$};
    \node at (3.5,0.7) {\large $a$};
    \node at (4.5,0.7) {\large $\#$};
    \node at (5.5,0.7) {\large $b$};
    \node at (6.5,0.7) {\large $a$};
    \node at (7.5,0.7) {\large $a$};
    \node at (8.5,0.7) {\large $\#$};
    \node at (9.5,0.7) {\large $b$};
    \node at (10.5,0.7) {\large $b$};
    \node at (11.5,0.7) {\large $\triangleleft$};
    \node at (1.5,1.4) {\large $\downarrow$};
    \node at (1.5,2) {\large $\Downarrow$};
  \end{tikzpicture}
\end{center}
The goal of the second pebble is to count the $\#$s in the input. The
transducer moves it forward until it reaches a $\#$, at which point it
pushes the third pebble $\triangledown$:
\begin{center}
  \begin{tikzpicture}
    \draw[thick] (0,0.4) -- (0,1) -- (12,1) -- (12,0.4) -- (0,0.4);
    \draw[thick] (1,0.4) -- (1,1);
    \draw[thick] (2,0.4) -- (2,1);
    \draw[thick] (3,0.4) -- (3,1);
    \draw[thick] (4,0.4) -- (4,1);
    \draw[thick] (5,0.4) -- (5,1);
    \draw[thick] (6,0.4) -- (6,1);
    \draw[thick] (7,0.4) -- (7,1);
    \draw[thick] (8,0.4) -- (8,1);
    \draw[thick] (9,0.4) -- (9,1);
    \draw[thick] (10,0.4) -- (10,1);
    \draw[thick] (11,0.4) -- (11,1);
    \node at (0.5,0.7) {\large $\triangleright$};
    \node at (1.5,0.7) {\large $a$};
    \node at (2.5,0.7) {\large $b$};
    \node at (3.5,0.7) {\large $a$};
    \node at (4.5,0.7) {\large $\#$};
    \node at (5.5,0.7) {\large $b$};
    \node at (6.5,0.7) {\large $a$};
    \node at (7.5,0.7) {\large $a$};
    \node at (8.5,0.7) {\large $\#$};
    \node at (9.5,0.7) {\large $b$};
    \node at (10.5,0.7) {\large $b$};
    \node at (11.5,0.7) {\large $\triangleleft$};
    \node at (1.5,2.6) {\large $\triangledown$};
    \node at (4.5,2) {\large $\Downarrow$};
    \node at (1.5,1.4) {\large $\downarrow$};
  \end{tikzpicture}
\end{center}
The goal of this third pebble $\triangledown$ is to copy the input block
that the first pebble $\downarrow$ points to. Therefore, it is going to move
forward while copying each input letter to the output, until it reaches
$\#$:
\begin{center}
  \begin{tikzpicture}
    \draw[thick] (0,0.4) -- (0,1) -- (12,1) -- (12,0.4) -- (0,0.4);
    \draw[thick] (1,0.4) -- (1,1);
    \draw[thick] (2,0.4) -- (2,1);
    \draw[thick] (3,0.4) -- (3,1);
    \draw[thick] (4,0.4) -- (4,1);
    \draw[thick] (5,0.4) -- (5,1);
    \draw[thick] (6,0.4) -- (6,1);
    \draw[thick] (7,0.4) -- (7,1);
    \draw[thick] (8,0.4) -- (8,1);
    \draw[thick] (9,0.4) -- (9,1);
    \draw[thick] (10,0.4) -- (10,1);
    \draw[thick] (11,0.4) -- (11,1);
    \node at (0.5,0.7) {\large $\triangleright$};
    \node at (1.5,0.7) {\large $a$};
    \node at (2.5,0.7) {\large $b$};
    \node at (3.5,0.7) {\large $a$};
    \node at (4.5,0.7) {\large $\#$};
    \node at (5.5,0.7) {\large $b$};
    \node at (6.5,0.7) {\large $a$};
    \node at (7.5,0.7) {\large $a$};
    \node at (8.5,0.7) {\large $\#$};
    \node at (9.5,0.7) {\large $b$};
    \node at (10.5,0.7) {\large $b$};
    \node at (11.5,0.7) {\large $\triangleleft$};
    \node at (4.5,2.6) {\large $\triangledown$};
    \node at (4.5,2) {\large $\Downarrow$};
    \node at (1.5,1.4) {\large $\downarrow$};
  \end{tikzpicture}\\
  Current output: $aba$
\end{center}
Now, in order to count the number of $\#$s -- which is the number of copies
of $aba$ that need to be outputted -- we would like to move the second
pebble $\Downarrow$ forward. The \emph{stack restriction} on pebble
transducers says that \emph{only the topmost pebble can be moved}.
Therefore, the transducer has to \emph{pop} the third pebble
$\triangledown$, \emph{forgetting its position}, so that $\Downarrow$
becomes free to move to the second $\#$:
\begin{center}
  \begin{tikzpicture}
    \draw[thick] (0,0.4) -- (0,1) -- (12,1) -- (12,0.4) -- (0,0.4);
    \draw[thick] (1,0.4) -- (1,1);
    \draw[thick] (2,0.4) -- (2,1);
    \draw[thick] (3,0.4) -- (3,1);
    \draw[thick] (4,0.4) -- (4,1);
    \draw[thick] (5,0.4) -- (5,1);
    \draw[thick] (6,0.4) -- (6,1);
    \draw[thick] (7,0.4) -- (7,1);
    \draw[thick] (8,0.4) -- (8,1);
    \draw[thick] (9,0.4) -- (9,1);
    \draw[thick] (10,0.4) -- (10,1);
    \draw[thick] (11,0.4) -- (11,1);
    \node at (0.5,0.7) {\large $\triangleright$};
    \node at (1.5,0.7) {\large $a$};
    \node at (2.5,0.7) {\large $b$};
    \node at (3.5,0.7) {\large $a$};
    \node at (4.5,0.7) {\large $\#$};
    \node at (5.5,0.7) {\large $b$};
    \node at (6.5,0.7) {\large $a$};
    \node at (7.5,0.7) {\large $a$};
    \node at (8.5,0.7) {\large $\#$};
    \node at (9.5,0.7) {\large $b$};
    \node at (10.5,0.7) {\large $b$};
    \node at (11.5,0.7) {\large $\triangleleft$};
    \node at (8.5,2) {\large $\Downarrow$};
    \node at (1.5,1.4) {\large $\downarrow$};
  \end{tikzpicture}\\
  Current output: $aba$
\end{center}
Next, $\triangledown$ is pushed again -- reinitialized to the first input
position -- in order to output a copy of $aba$ as before. After this,
$\triangledown$ is popped and $\Downarrow$ moves forward, until it reaches
the end of the word:
\begin{center}
  \begin{tikzpicture}
    \draw[thick] (0,0.4) -- (0,1) -- (12,1) -- (12,0.4) -- (0,0.4);
    \draw[thick] (1,0.4) -- (1,1);
    \draw[thick] (2,0.4) -- (2,1);
    \draw[thick] (3,0.4) -- (3,1);
    \draw[thick] (4,0.4) -- (4,1);
    \draw[thick] (5,0.4) -- (5,1);
    \draw[thick] (6,0.4) -- (6,1);
    \draw[thick] (7,0.4) -- (7,1);
    \draw[thick] (8,0.4) -- (8,1);
    \draw[thick] (9,0.4) -- (9,1);
    \draw[thick] (10,0.4) -- (10,1);
    \draw[thick] (11,0.4) -- (11,1);
    \node at (0.5,0.7) {\large $\triangleright$};
    \node at (1.5,0.7) {\large $a$};
    \node at (2.5,0.7) {\large $b$};
    \node at (3.5,0.7) {\large $a$};
    \node at (4.5,0.7) {\large $\#$};
    \node at (5.5,0.7) {\large $b$};
    \node at (6.5,0.7) {\large $a$};
    \node at (7.5,0.7) {\large $a$};
    \node at (8.5,0.7) {\large $\#$};
    \node at (9.5,0.7) {\large $b$};
    \node at (10.5,0.7) {\large $b$};
    \node at (11.5,0.7) {\large $\triangleleft$};
    \node at (11.5,2) {\large $\Downarrow$};
    \node at (1.5,1.4) {\large $\downarrow$};
  \end{tikzpicture}\\
  Current output: $abaaba$
\end{center}
When this happens, the transducer pops $\Downarrow$ and moves $\downarrow$
forward to the beginning of the second block $baa$. Then $\Downarrow$ is pushed again
to count the $\#$s:
\begin{center}
  \begin{tikzpicture}
    \draw[thick] (0,0.4) -- (0,1) -- (12,1) -- (12,0.4) -- (0,0.4);
    \draw[thick] (1,0.4) -- (1,1);
    \draw[thick] (2,0.4) -- (2,1);
    \draw[thick] (3,0.4) -- (3,1);
    \draw[thick] (4,0.4) -- (4,1);
    \draw[thick] (5,0.4) -- (5,1);
    \draw[thick] (6,0.4) -- (6,1);
    \draw[thick] (7,0.4) -- (7,1);
    \draw[thick] (8,0.4) -- (8,1);
    \draw[thick] (9,0.4) -- (9,1);
    \draw[thick] (10,0.4) -- (10,1);
    \draw[thick] (11,0.4) -- (11,1);
    \node at (0.5,0.7) {\large $\triangleright$};
    \node at (1.5,0.7) {\large $a$};
    \node at (2.5,0.7) {\large $b$};
    \node at (3.5,0.7) {\large $a$};
    \node at (4.5,0.7) {\large $\#$};
    \node at (5.5,0.7) {\large $b$};
    \node at (6.5,0.7) {\large $a$};
    \node at (7.5,0.7) {\large $a$};
    \node at (8.5,0.7) {\large $\#$};
    \node at (9.5,0.7) {\large $b$};
    \node at (10.5,0.7) {\large $b$};
    \node at (11.5,0.7) {\large $\triangleleft$};
    \node at (4.5,2) {\large $\Downarrow$};
    \node at (5.5,1.4) {\large $\downarrow$};
  \end{tikzpicture}\\
  Current output: $abaaba\#$
\end{center}
At this point, the third pebble $\triangledown$ is pushed again, its purpose
being to copy the block $baa$ to the output. To do so, it moves forward until it
reaches the beginning of that block, pointed by $\downarrow$:
\begin{center}
  \begin{tikzpicture}
    \draw[thick] (0,0.4) -- (0,1) -- (12,1) -- (12,0.4) -- (0,0.4);
    \draw[thick] (1,0.4) -- (1,1);
    \draw[thick] (2,0.4) -- (2,1);
    \draw[thick] (3,0.4) -- (3,1);
    \draw[thick] (4,0.4) -- (4,1);
    \draw[thick] (5,0.4) -- (5,1);
    \draw[thick] (6,0.4) -- (6,1);
    \draw[thick] (7,0.4) -- (7,1);
    \draw[thick] (8,0.4) -- (8,1);
    \draw[thick] (9,0.4) -- (9,1);
    \draw[thick] (10,0.4) -- (10,1);
    \draw[thick] (11,0.4) -- (11,1);
    \node at (0.5,0.7) {\large $\triangleright$};
    \node at (1.5,0.7) {\large $a$};
    \node at (2.5,0.7) {\large $b$};
    \node at (3.5,0.7) {\large $a$};
    \node at (4.5,0.7) {\large $\#$};
    \node at (5.5,0.7) {\large $b$};
    \node at (6.5,0.7) {\large $a$};
    \node at (7.5,0.7) {\large $a$};
    \node at (8.5,0.7) {\large $\#$};
    \node at (9.5,0.7) {\large $b$};
    \node at (10.5,0.7) {\large $b$};
    \node at (11.5,0.7) {\large $\triangleleft$};
    \node at (5.5,2.6) {\large $\triangledown$};
    \node at (4.5,2) {\large $\Downarrow$};
    \node at (5.5,1.4) {\large $\downarrow$};
  \end{tikzpicture}\\
  Current output: $abaaba\#$
\end{center}
The transducer can know that it has reached the beginning of the block it wants
to copy because it is allowed to \emph{compare the positions} of $\triangledown$
and $\downarrow$. The third pebble $\triangledown$ can now move forward to
output $baa$, and the execution continues as the reader might expect.

\paragraph{Moral of the story}

Note that every time an output letter is produced, the first pebble $\downarrow$
must be on the leftmost position of the block in which the third pebble
$\triangledown$ is. Thus, the output positions can be morally parameterized by
pairs of input positions, corresponding to $\Downarrow$ and $\triangledown$;
this is why the transducer has only quadratic growth, despite its 3 pebbles.
(This idea will make an appearance again in Example~\ref{exa:innsq-interp}.) The
first pebble $\downarrow$ is redundant, but this redundancy is made necessary by
the stack condition: the role of $\downarrow$ could be seen as keeping some
partial information about $\triangledown$ while $\Downarrow$ moves, since to be
allowed to move $\Downarrow$, first $\triangledown$ must be popped.

\section{A simple proof that inner squaring requires 3 pebbles}
\label{sec:innsq-non-min}

The 3-pebble transducer described in the previous section computes the inner
squaring function $\innsq$ (Example~\ref{exa:innsq}); we will see in
Example~\ref{exa:innsq-pebble} how to turn this into a more formal argument for
the fact that $\innsq \in \Pebble_3$. The goal of this section is to reprove
{\cite[Theorem~6.3]{PolyregSurvey}}:
\begin{theorem}
  \label{thm:counterexample}
  Inner squaring cannot be computed with 2 pebbles: $\innsq \notin \Pebble_2$.
\end{theorem}
\begin{remark}
  However, one can show that the restriction of $\innsq$ to inputs in $\{a,\#\}^*$ is in $\Pebble_2$.  
\end{remark}
\begin{remark}\label{rem:deatomization}
  Bojańczyk's proof sketch for~\cite[Theorem~6.3]{PolyregSurvey} actually shows that the function
  $a_1 \dots a_n \in \atoms^* \mapsto (a_1)^n \dots (a_n)^n$
  cannot be computed by a 2-pebble \emph{atom-oblivious} transducer (here the alphabet $\atoms$ is an \emph{infinite} set of \emph{atoms}). Combining this with the Deatomization Theorem from~\cite{PolyregGrowth} -- whose proof is rather complicated -- shows that no $f \colon \{\langle,\rangle,\bullet\}^* \to \{\langle,\rangle,\bullet\}^*$ in $\Pebble_2$ can satisfy
  $f(\langle \bullet^{p_1} \rangle \dots \langle \bullet^{p_n} \rangle) = (\langle \bullet^{p_1} \rangle)^n \dots (\langle \bullet^{p_n} \rangle)^n$,
  from which Theorem~\ref{thm:counterexample} can be deduced using the composition properties of pebble transducers~\cite{PebbleComposition}.
\end{remark}

Instead of introducing the hierarchy $(\Pebble_k)_{k\in\bbN}$ by a concrete
machine model, we define it in \Cref{sec:pebble-def} using
combinators (operators on functions). This abstract presentation depends on the
\emph{origin semantics}~\cite{origin} (see also~\cite[\S5]{MuschollPuppis}) of
regular functions, about which we say a few words in \Cref{sec:reg-origins}. We
also recall in~\S\ref{sec:pebble-def} a similar definition of polyblind
functions; their only use in this section will be to state and prove
Corollary~\ref{cor:innsq-not-polyblind}, but they will appear again briefly in
\Cref{sec:fo-interpretations}. After all this, \Cref{sec:innsq-proof} proves
Theorem~\ref{thm:counterexample}.

\subsection{Regular functions with origin information}
\label{sec:reg-origins}

Regular functions are those computed by 1-pebble transducers, also known as
\emph{two-way transducers} (2DFTs). We will mostly avoid explicit manipulations
of 2DFTs; this is why we do not recall a formal definition here, and refer to the
survey~\cite{MuschollPuppis} for more on regular functions and their origin
semantics. A typical example of a regular string-to-string function is:
\[ a^{m_0}\# \dots \# a^{m_n} \in \{a,\#\}^* \mapsto a^{m_n} b^{m_n} \# \dots \# a^{m_0} b^{m_0} \]
There are several natural ways to lift it to a regular function with origins, which reflect different ways to compute it with a 2DFT: for example $aaa\#aa$ could be mapped to
\begin{align*}
  \worig{a}{5}\worig{a}{6}\worig{b}{5}\worig{b}{6}
  \worig{\#}{4}
  \worig{a}{1}\worig{a}{2}\worig{a}{3}\worig{b}{1}\worig{b}{2}\worig{b}{3}
  \quad\text{or}\quad
  \worig{a}{5}\worig{a}{6}\worig{b}{6}\worig{b}{5}
  \worig{\#}{4}
  \worig{a}{1}\worig{a}{2}\worig{a}{3}\worig{b}{3}\worig{b}{2}\worig{b}{1}
  \end{align*}
The second component indicates, for each output letter, which input position it
\enquote{comes from} -- that is, where the reading head was placed when the
letter was outputted.

We shall write $f^\circ,g^\circ,\ldots \colon \Gamma^* \to (\Sigma\times\bbN)^*$ for regular functions with origin information and $f,g,\ldots \colon \Gamma^* \to \Sigma^*$ for the corresponding regular string-to-string functions. Thus, $f = (\pi_1)^* \circ f^\circ$ and $(\pi_2)^* \circ f(w) \in \{1,\dots,|w|\}^*$ for any $w\in\Gamma^*$, where $\pi_i$ is the projection on the $i$-th component of a pair ($i \in \{1,2\}$).

\begin{remark}
  The notion of origin semantics can be a bit problematic when the empty string
  $\varepsilon$ has a non-empty image. This is why, in this section, we consider
  that \textbf{all our (poly)regular functions map $\varepsilon$ to
    $\varepsilon$} in order to avoid inessential inconveniences. This makes no
  difference concerning the strength of our result, since
  $\innsq(\varepsilon)=\varepsilon$.
\end{remark}

\subsection{Pebble transducers in an abstract style and polyblind functions}
\label{sec:pebble-def}

We denote by $\underline\Sigma$ a disjoint copy of $\Sigma$,
made of \enquote{underlined letters}, so that $a \in \Sigma \mapsto
\underline{a} \in \underline\Sigma$ is a bijection. We also write
$w\withul i$ for $w[1]\ldots w[i-1] \underline{w[i]} w[i+1]\ldots w[n] \in
(\Sigma \cup \underline\Sigma)^*$ where $n = |w|$.

\begin{definition}
  Let $f^\circ \colon \Gamma^* \to (I\times\bbN)^{*}$ be a regular function with origins
 (so $f: \Gamma^* \to I^{*}$ is regular) and, for $i \in I$, let $g_{i} \colon (\Gamma \cup \underline\Gamma)^{*} \to \Sigma^{*}$ and $h_i \colon \Gamma^* \to \Sigma^*$. For $w \in \Gamma^*$, define
  \begin{align*}
     \pebble(f^\circ,(g_{i})_{i\in I})(w) &= g_{i_{1}}(w\withul j_{1}) \cdot \ldots \cdot g_{i_{n}}(w\withul j_{n}) \quad\mathrm{where}\quad f^\circ(w) = (i_{1},j_{1}) \dots (i_{n},j_{n})\\
     \blind(f,(h_{i})_{i\in I})(w) &= h_{i_{1}}(w) \cdot \ldots \cdot h_{i_{n}}(w)
  \end{align*}
  ($\mathsf{blind}$ is called \enquote{composition by substitution} in~\cite[\S4]{NguyenNP21}). Using these combinators, we define the hierarchies of function classes $\Pebble_n$ and $\Blind_n$ inductively.
  $\Pebble_0 = \Blind_0$ is the class of string-to-string functions with finite range (or equivalently, whose output has bounded length), and
  \begin{align*}
    \forall n\in\bbN,\quad \Pebble_{k+1} &= \{ \pebble(f^\circ, (g_i)) \mid f^\circ\ \text{regular},\; g_i \in \Pebble_k\}\\
     \Blind_{k+1} &= \{  \blind(f, (g_i)) \mid f\ \text{regular},\; g_i \in \Blind_k\}
  \end{align*}
\end{definition}

Morally, the correspondence with the informal presentation of $k$-pebble
transducers in \Cref{sec:innsq} is that we can compute
$\pebble(f^\circ,(g_{i})_{i\in I})$ by tweaking some two-way transducer
$\mathcal{T}$ computing $f^\circ$: each time $\mathcal{T}$ would output a letter
$i$, we instead \enquote{call a subroutine} computing $g_i(w\withul j)$ where
$j$ is the current position of the reading head. If every $g_i$ is implemented
using a stack of $k$ pebbles, then we can implement the subroutine-calling
machine with a stack of height $k+1$: the pebble at the bottom of the stack
corresponds to the reading head of $\mathcal{T}$. Conversely, in any pebble
transducer, pushing a pebble can be seen as initiating a subroutine call.

\begin{example}\label{exa:innsq-pebble}
  Let us turn the 3-pebble transducer of
  \Cref{sec:innsq} into an \enquote{abstract-style} proof that
  $\innsq\in\Pebble_3$.
  Let $f^\circ \colon \{a,b,\#\}^* \to (\{\bullet,\#\}\times\bbN)^*$ be defined by
  \[ f^\circ(\underbrace{w[1] \dots w[i_1-1]}_{\mathclap{\text{each block is in}\ \{a,b\}^*}} \# \dots \# w[i_m+1] \dots w[n]) = \worig{\bullet}{1}\worig{\#}{i_1}\underbrace{\worig{\bullet}{i_1+1}}_{\mathclap{\text{each $\bullet$ is omitted if the corresponding input block is empty}}}\dots\worig{\#}{i_m}\worig{\bullet}{i_m+1} \]
  and $h(\dots\#\underline{c}w\#\dots) = cw$ for $c\in\{a,b\}$ and $w\in\{a,b\}^*$. Then
  \[ \underbrace{\innsq}_{\mathclap{\text{Example~\ref{exa:innsq}}}} = \underbrace{\pebble\big(\overbrace{f^\circ}^{\mathclap{\text{regular with origin information}}}, (g_i)_{i\in\{\bullet,\#\}}\big)}_{\in\Pebble_3} \qquad g_\bullet = \underbrace{\blind\big(\overbrace{w \mapsto \bullet^{|w|_\#}, (h)_{j\in\{\bullet\}}}^{\mathclap{\text{both regular i.e.\ in}\ \Pebble_1}}\big)}_{\in\Blind_2\subset\Pebble_2}, \; g_\# \colon \underbrace{w \mapsto \#}_{\mathclap{\in\Pebble_0}} \]
\end{example}
The fact that $\Blind_k$ consists of the functions computed by blind $k$-pebble
transducers -- that \emph{cannot compare the positions of their pebbles} -- is stated
in~\cite[\S5]{NguyenNP21}, with a detailed proof available
in~\cite[Appendix~D]{NguyenNP21} refining the intuitions on
\enquote{subroutines} given above. By a straightforward adaptation of that
proof:
\begin{proposition}
  $\Pebble_k$ is exactly the class of string-to-string functions computed by \emph{$k$-pebble transducers}~\cite[\S2]{PolyregSurvey} for any $k\geq1$.
\end{proposition}
Note that a similar definition of the \enquote{last pebble} transducers
mentioned in the introduction is given in~\cite[Definition~3.3]{LastPebble},
without explicit proof that they coincide with the expected machine model.

Finally, let us recall that unlike general pebble transducers,
blind pebble transducers enjoy the pebble minimization
property~\cite[Theorem~7.1]{NguyenNP21}. See also~\cite{LastPebble} for an
\emph{effective} minimization proof\footnote{It is likely that the proof
  of~\cite{NguyenNP21} could be made effective, but this is not explicit.}
(i.e.\ that provides an algorithm to find, given a blind $\ell$-pebble
transducer, its degree $k$ of growth and an equivalent blind $k$-pebble
transducer) and a generalization to \enquote{last pebble} transducers.

\begin{theorem}
  \label{thm:blind-min}
  For all $k\in\bbN$, $\Blind_k = \left\{f \in \displaystyle\bigcup_{\ell\in\bbN} \Blind_\ell \;\middle|\; |f(w)| = O(|w|^k)\right\}$.
\end{theorem}

\begin{corollary}
  \label{cor:innsq-not-polyblind}
  The inner squaring function (Example~\ref{exa:innsq}) is not polyblind.
  \end{corollary}
  \begin{proof}
    Since it has quadratic growth, if it were polyblind, it would be in $\Blind_2\subset\Pebble_2$. This would contradict the main result of this section (Theorem~\ref{thm:counterexample}).
\end{proof}

\subsection{Proof of Theorem~\ref{thm:counterexample}}
\label{sec:innsq-proof}

Our approach to prove $\innsq\notin\Pebble_2$ goes through the output languages of regular functions.
\begin{definition}
Call a language $L \subseteq \Sigma^*$ a \emph{regular image} if there exists a regular function $f$ with codomain $\Sigma^*$ such that $L = \Image(f)$.
\end{definition}
Not all regular images are regular or even context-free languages: consider e.g.\ $\{ a^n b^n c^n d^n \mid n \in \mathbb{N}\}$.

We will show that no function in $\Pebble_2$ can coincide with $\innsq$ on the subset of inputs
\[ (a^*b\#)^* \#^* = \{\text{strings with the shape}\ a\dots ab\# \dots \# a\dots ab \#\# \dots \#\} \]
For the sake of contradiction, assume the opposite.
\begin{claim}\label{clm:reg-image-right-shape}
  Under this assumption, there exists a language $L \subseteq b\{a,b\}^*b$
  \begin{itemize}
    \item which is a regular image
    \item that contains only factors of words from $\innsq((a^*b\#)^*\#^*)$
    \item and such that, for any $N\in\bbN$, some word in $L$ has the shape
      $\overbrace{b\underbrace{a \dots a}_{\mathclap{\text{every block of $a$s
              has the same length}\ n\geq N}} b \dots b\underbrace{a \dots a}b}^{\mathclap{\text{there are at least $N$ $b$s}}}$.
  \end{itemize}
\end{claim}
\begin{proof}
  We start by an analysis of some output factors arising in the computation of $\innsq$ by a hypothetical 2-pebble transducer, which will later inform the
  construction of a suitable language $L$. Unfolding our formal definition of $\Pebble_2$, our assumption means that $\innsq$ coincides on $(a^*b\#)^*\#^*$
  with some function $\pebble(g^\circ,(h_i)_{i\in I})$ where
  \begin{itemize}
    \item $g^\circ \colon \{a,b,\#\}^{*} \to (\{a,b,\#\}\times\bbN)^{*}$ is a regular function with origins;
    \item for $i\in I$, the function $h_{i} \colon \{a,b,\#,\underline{a},\underline{b},\underline\#\}^{*} \to \{a,b,\#\}^*$ is in $\Pebble_1$, so it is regular.
  \end{itemize}
  The functions $g^\circ$ and $h_i$, being regular, have linear growth: there is
  some constant $C \in \mathbb{N}$ such that, for all input words $u,v$ and $i
  \in I$, we have $|g^\circ(u)| \le C|u|$ and $|h_i(v)| \le C|v|$.

  Consider the input word $u = (a^n b\#)^n \#^{nm}$ for some $n,m\in\bbN$ (which shall be taken large enough to satisfy constraints that will arise during the proof). By definition,
  \[ \innsq(u) = v_1 \dots v_k \quad
  \text{where} \quad g^\circ(w) = (i_{1}, j_1) \dots (i_{k},j_{k})
  \quad\text{and}\quad v_\ell = h_{i_{\ell}}(u\withul j_{\ell}) \]
    For large enough $n$ and $m$, we have that for all $\ell \in \{1,\dots,k\}$,
\[ |v_\ell| \leq C|u\withul j_\ell| = Cn(n+2+m) < (n+1)n(1+m) =
\underbrace{|a^n b| \times |u|_\#}_{\mathclap{\text{distance between two consecutive non-adjacent occurrences of $\#$ in $\innsq(u)$}\qquad\qquad\qquad\qquad}} \]
Thanks to the above, if we call $p$ the largest integer such that $|v_1 \dots v_p|_\# < n$, then:
\begin{itemize}
  \item $|v_\ell|_\# \leq 1$ for every $\ell \in \{1,\dots,p\}$;
  \item $v_1 \dots v_p$ is a prefix of $\innsq(u)$ that contains all the $n-1$ first occurrences of $\#$.
\end{itemize}
As a consequence of the latter item,
$((a^n b)^{n(1+m)})\#)^{n-1}$ is a prefix of $v_1 \dots v_p$ and 
  $|v_1 \dots v_p|_b \ge n(1+m)(n-1)$. Meanwhile, we also have $p < k = |g^\circ(u)| \leq C|u| = Cn(n + 2 + m)$. Hence
\[ \frac{|v_1 \dots v_p|_b}{p} \geq \frac{n(1+m)(n-1)}{Cn(n+2+m)} \xrightarrow[~n,m \to +\infty~]{} +\infty \]
Thus, by the pigeonhole principle, for an input $u$ parameterized by large
enough values of $n$ and $m$, there is at least one output factor $v_\ell$
containing $2N$ or more occurrences of $b$. Since $|v_\ell|_\# \leq 1$, either
its largest $\#$-free prefix or its largest $\#$-free suffix (or both) contains
$N$ or more occurrences of $b$. After applying the regular function $a\dots abwba \dots a \mapsto bwb$ (for $w \in \{a,b\}^*$) to trim this prefix or suffix, we get a factor of $v_\ell$ -- and therefore of $\innsq(u)$ -- in $b(a^n b)^*$, and we may take $n \geq N$ to ensure that the blocks of $a$s have length at least $N$.

To sum up, for every $N$, there are some $n,r \geq N$ such that
$b(a^n b)^r$ belongs to the following language $L$,
where $p_i$ is the identity on $\{u\withul j \mid u\in (a^*b\#)^*\#^*,\,
(i,j)\ \text{appears in}\ g^\circ(u)\}$ and sends every other word of
$\{a,b,\#,\underline{a},\underline{b},\underline\#\}^{*}$ to
$\varepsilon$:
\[ L = \bigcup_{i\in I} \Image(\text{trim} \circ \text{largest \#-free
    prefix} \circ h_i \circ p_i) \cup \Image(\text{trim} \circ \text{largest
    \#-free suffix} \circ h_i \circ p_i) \]
The \enquote{projection} $p_i$ is here to make sure that any $v \in \Image(h_i
\circ p_i)$ is a factor of some word from $\innsq((a^*b\#)^*\#^*)$ --
which, in turn, guarantees that $L$ itself contains only such factors.

To complete the proof, we need to show that $L$ is a regular image. One can
check that regular images are closed under finite unions, so it is enough that
$\Image(-)$ is applied to regular functions in the above expression. That is
indeed the case: regular functions are closed under
composition\footnote{This is a quite non-trivial result, whose use could be avoided
in our proof at the price of more explicit manipulations of two-way
transducers.}~\cite{ChytilJ77}, and $p_i$ is regular because it is computed by a
two-way transducer that first checks in a single pass that its input is in
$(a^*b\#)^*\#^*$ when the underlining is ignored, then simulates the transducer
for $g^\circ$ -- without producing output -- until it observes that some $i$ is
indeed outputted with origin at the underlined position, and finally, if the
previous checks have not failed, copies its input.
\end{proof}

We shall now use the following pumping lemma due to Rozoy~{\cite[\S4.1]{Rozoy}}:
\begin{lemma}
\label{lem:2dftoutpump}
  If $L$ is a regular image, then for some $k,K \in \bbN$, every $w \in L$ with $|w| \geq K$ has a
  decomposition $w = u_{0} v_{1} \dots u_{k-1} v_{k} u_{k}$ with
   \begin{itemize}
     \item $\exists i \in \{1,\dots,k\} \colon v_{i} \neq \varepsilon$
     \item $\forall i \in \{1,\dots,k\},\, |v_{i}| \leq K$
     \item $\{u_{0} (v_{1})^{n} \dots u_{k-1} (v_{k})^{n} u_{k} \mid n \in\bbN\} \subseteq L$
   \end{itemize}
\end{lemma}
Let us apply this lemma to the language $L$ given by Claim~\ref{clm:reg-image-right-shape}, yielding two constants $k,K\in\bbN$. One of the properties stated in Claim~\ref{clm:reg-image-right-shape} is that there exist $n \geq K$ and $r \geq 2k+1$ such that $b(a^nb)^r \in L$. This string has length greater than $K$, so it is pumpable:
\[ b(a^nb)^r = u_0 v_1 \dots u_{k-1} v_k u_k \qquad\text{and}\qquad w = u_0 (v_1)^2 \dots u_{k-1} (v_k)^2 u_k \in L \]
Let us now show that $ba^nb$ is a factor of $w$. Since $|v_i| \leq K \leq n$ for every $i \in \{1,\dots,k\}$, and any two occurrences of $b$ in $b(a^nb)^r$ are separated by $a^n$, each factor $v_i$ can contain at most one $b$. So $|v_1|_b + \dots + |v_k|_b \leq k$ which leads to $|u_0|_b + \dots + |u_k|_b \geq r+1-k \geq k+2$. One of the $u_j$ must then contain two occurrences of $b$, and this $u_j$ is also a factor of $w$.

Recall that $L$ is comprised exclusively of $\#$-free factors of words in $\innsq((a^*b\#)^*\#^*)$, so \emph{all the \enquote{blocks of $a$s} in $w$ must have the same size}. Having seen above that this size must be $n$,
we can now wrap up our proof by contradiction of
Theorem~\ref{thm:counterexample} with a case analysis:
\begin{itemize}
  \item First assume that $|v_i|_b \geq 1$ for some $i \in \{1,\dots,k\}$. We may write $v_i = a^\ell b \dots b a^m$ (where the first and last $b$ can coincide) in this case. Then $w$ contains $(v_i)^2$ which in turn contains a factor $ba^{\ell+m}b$. So $\ell+m=n$ since it is the size of a block of $a$s in $w$; but at the same time, using the second item of Lemma~\ref{lem:2dftoutpump}, $\ell+m < |v_i| \leq K \leq n$.
  \item Otherwise, $v_i \in \{a\}^*$ for all $i$, so pumping does not increase the number of $b$s. Therefore $w$ has as many blocks of $a$s as $b(a^n b)^r$, and we have also seen that its blocks have size $n$, so $w = b(a^n b)^r$. But since at least one $v_i$ is nonempty, $|w| > |b(a^n b)^r|$.
\end{itemize}

\begin{remark}\label{rem:iterative}
  Rozoy proves Lemma~\ref{lem:2dftoutpump} for deterministic two-way transducers (cf.~\Cref{sec:reg-origins}). She also shows~\cite[\S4.2]{Rozoy} that the output languages of \emph{nondeterministic} two-way transducers enjoy a weaker version of the lemma without the bound on the length of the $v_i$ (this bound is refuted by the example $\{(w\#)^n \mid w \in \Sigma^*,\; n \in \mathbb{N}\}$) (a result later rediscovered by Smith~\cite{Pumping2NFT}); in general, the languages that satisfy this weaker pumping lemma are called \emph{$k$-iterative} in the literature (see e.g.~\cite{Pumping2NFT,KanazawaKMSY14}). For more on regular images, see~\cite{EhrichY71,EngelfrietH91}; several references are also given in~\cite[p.~18]{GauwinHDR}.\footnote{Instead of Latteux's unobtainable
  technical report cited by~\cite{GauwinHDR}, see~\cite[Prop.~I.2]{Latteux79} \& compare with~\cite{Rajlich72}.}
\end{remark}

\section{First-order interpretations (in 2 dimensions)}
\label{sec:fo-interpretations}

In this section, which is entirely independent from the previous one, we
recall another way to specify string-to-string functions, namely first-order
(FO) interpretations. There are two main motivations:
\begin{itemize}
\item they are a convenient tool (but not strictly necessary, cf.~Remark~\ref{rem:fo-useless}) for the next section;
\item we wish to state and prove Theorem~\ref{thm:polyblind-refutation},
  refuting~\cite[Conjecture~10.1]{NguyenNP21}.
\end{itemize}
We assume basic familiarity with first-order logic.
A word $w\in\Sigma^*$ can be seen as a structure whose domain is the set of \emph{positions} $\{1,\dots,|w|\}$, over the relational vocabulary consisting of:
\begin{itemize}
  \item for each $c\in\Sigma$, a unary symbol $c$ where $c(i)$ is interpreted as true whenever $w[i]=c$,
  \item and a binary relation symbol $\leq$, interpreted as the total order on positions.
\end{itemize}
As an example, for $\Sigma=\{a,b\}$, let $F(x) = b(x) \lor \forall y.\; (a(y) \lor x \leq y)$. A string $w\in\Sigma^*$ satisfies $F(i)$ for a given $i\in\{1,\dots,w\}$ -- notation: $w \models F(i)$ -- when either $w[i]=b$, or the position $i$ contains an $a$ which occurs before (at the left of) all the $b$s in the word $w$. Thus, the formula $\forall x.\; F(x)$ evaluates to true exactly over the words in $a^*b^*$.

This model-theoretic perspective also leads to a way to specify string-to-string functions.
\begin{definition}
  \label{def:fo-interp}
  Let $k\geq1$ and $\Gamma,\Sigma$ be alphabets.
  A \emph{two-dimensional first-order interpretation} $\cI$ from $\Gamma^*$ to
  $\Sigma^*$ consists of several FO formulas over $\Gamma$: $\cI_c(x_1,x_2)$
  for each $c\in\Sigma$, and $\cI_\leq(x_1,x_2,y_1,y_2)$.

  For $u\in\Gamma^*$, let $O^\cI_u$ be the structure
  \begin{itemize}
  \item with domain $\{ \langle i_1,i_2 \rangle \in \{1,\dots,|u|\}^2 \mid \exists c\in\Sigma \colon u \models \cI_c(i_1,i_2) \}$;
    \item where $c(\langle i_1,i_2\rangle)$ iff $u \models \cI_c(i_1,i_2)$, and $\langle i_1,i_2\rangle \leq\langle j_1,j_2\rangle$ iff $u\models\cI_\leq(i_1,i_2,j_1,j_2)$.
  \end{itemize}
  Then $\cI$ defines the function
  $u \in \Gamma^* \mapsto \displaystyle\begin{cases}
                                         v & \text{if}\ O^\cI_u \cong \text{the structure corresponding to}\ v\\
                                         \varepsilon & \text{when\footnotemark{} there is no such}\ v \in \Sigma^*
                                       \end{cases}$\\
                                       \footnotetext{This is purely for
                                         convenience, to avoid having to consider
                                         partial functions. The language of input words $u$ for
                                         which such a $v$ exists is first-order definable, so a partial FO interpretation can always be completed to a total one.}

\end{definition}
\begin{example}\label{exa:interp0}
  $a^n \in\{a\}^* \mapsto (a^{n-1}b)^{n-1}$ is defined by a two-dimensional FO interpretation $\cI$ such that $\cI_\leq$ is the lexicographic order over pairs and
  \[ \cI_a(x_1,x_2) = \lnot\underbrace{\max(x_1)}_{\mathclap{\text{i.e.}\ \forall y.\; y \leq x_1}} \land \lnot\max(x_2) \qquad \cI_b(x_1,x_2) = \lnot\max(x_1) \land \max(x_2) \]
\end{example}
\begin{example}\label{exa:innsq-interp}
  As a more subtle example, the inner squaring function (Example~\ref{exa:innsq})
  admits a two-dimensional FO interpretation that we intuitively illustrate over
  the input $aba\#baa\# bb$ as follows:
  \begin{center}
  \includegraphics{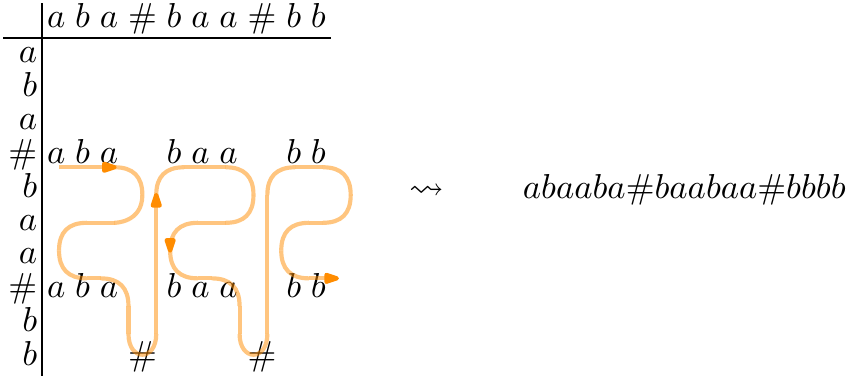}
  \end{center}
  The array on the left handside indicates the output letter associated
  with each pair of positions in the input word (if they are in the domain)
  and the path represents the order in which these coordinates should be read
  to form the inputs. The interpretation is defined by the following formulas:
   \begin{itemize}
      \item $\cI_{a}(x_{1},x_{2}) = a(x_{1}) \land \#(x_{2})$ and $\cI_{b}(x_{1},x_{2}) = b(x_{1}) \land \#(x_{2})$
      \item $\cI_{\#}(x_{1},x_{2}) = \#(x_1) \land \max(x_2)$
      \item $\cI_{\leq}(x_{1},x_{2},y_{1},y_{2}) =$ either $\#(y_1) \land x_1 \leq y_1$, or there exist $x_3,y_3$ such that
      \begin{itemize}
        \item $x_3 \leq x_1$, and there are no $\#$s strictly in-between $x_3$ and $x_1$;
        \item $y_3 \leq y_1$, and there are no $\#$s strictly in-between $y_3$ and $y_1$;
        \item neither $x_3$ nor $y_3$ has an immediate predecessor which is an $a$ or a $b$;
        \item $(x_3,x_2,x_1) \leq (y_3,y_2,y_1)$ for the lexicographic order on 3-tuples.
      \end{itemize}
      Note the connection with the 3-pebble transducer of \Cref{sec:innsq}: $x_3,x_2,x_1$ correspond to the positions of the respective pebbles $\downarrow,\Downarrow,\triangledown$.
  \end{itemize}
\end{example}

\begin{theorem}[{Bojańczyk, Kiefer \& Lhote~\cite{msoInterpretations}}]\label{thm:interp-polyreg}
  Every function specified by a first-order interpretation is polyregular, i.e.\
  in $\displaystyle\bigcup_{k\in\bbN}\Pebble_k$.
\end{theorem}

This holds for interpretations of arbitrary dimension $k\in\bbN$, even though we
only defined the $k=2$ case; and also for MSO interpretations, where first-order
logic is replaced by monadic second-order logic. A converse is also shown
in~\cite{msoInterpretations}: every polyregular function coincides, on inputs of
length at least 2, with an MSO interpretation.
\begin{remark}\label{rem:mso-interp}
  One can drop the condition \enquote{of length at least 2} using the
  \enquote{multi-component} variant of MSO
  interpretations~\cite[\S5.1]{PolyregSurvey}. These also have the more
  significant advantage that a polyregular function of growth $O(n^k)$ can be
  expressed as a $k$-dimensional interpretation~\cite[\S2]{PolyregGrowth}: the
  natural counterpart to pebble minimization works for multi-component MSO
  interpretations.
\end{remark}

Similarly to the characterization of polyregular functions by MSO
interpretations, it would be desirable to have a logical formalism for polyblind
functions (those that are in $\Blind_k$ for some $k$). Unfortunately -- and this
is a new result:
\begin{theorem}
  \label{thm:polyblind-refutation}
  The characterization of polyblind functions conjectured in~\cite[\S10]{NguyenNP21} is wrong.
\end{theorem}
\begin{proof}
For purely syntactic reasons, the two-dimensional FO interpretations such that:
\begin{itemize}
  \item in every formula defining the interpretation, each variable can be given a sort in $\{1,2\}$ in such a way that variables of different sorts are never compared (for equality or ordering)
  \item and in $\cI(x_1,x_2,y_1,y_2)$ and $\cI_c(x_1,x_2)$ ($c\in\Sigma$), $x_i$ and $y_i$ have sort $i\in\{1,2\}$
\end{itemize}
 can be seen as a special case of the logical interpretations considered in~\cite[Conjecture~10.1]{NguyenNP21}. Example~\ref{exa:innsq-interp} fits these criteria (with $x_3$ and $y_3$ having sort 1). If the conjecture were true, $\innsq$ would therefore be polyblind, contradicting Corollary~\ref{cor:innsq-not-polyblind}.
\end{proof}

\section{Quadratic polyregular functions vs macro tree transducers}
\label{sec:engelfrieteries}

Our goal now is to show the following:
\begin{theorem}\label{thm:engelfrieteries}
  For any $k\geq1$, there exists a string-to-string function $f_k$ such that:
  \begin{itemize}
    \item $f_k$ is computed by a two-dimensional first-order interpretation (see
      the previous section);
    \item for any $k$-tuple of functions $(g_1, \dots, g_k)$ such that each $g_i$ is computed by some \emph{macro tree transducer}, $\Image(g_1 \circ \dots \circ g_k) \neq \Image(f_k)$.
  \end{itemize}
\end{theorem}
In the second item, the domain of $g_i$ must be equal to the codomain of $g_{i+1}$ for the composition to be well-defined; and to make the comparison of output languages meaningful, the codomains of $g_1$ and $f_k$ should be equal.
But $g_1$ outputs ranked trees, whereas $f_1$ outputs strings. To make sense of this, as usual, we identify $\Sigma^*$ with the set of trees over the ranked alphabet that consists of a letter $\widehat{c}$ of rank 1 for each $c\in\Sigma$, plus a single letter $\widehat{\varepsilon}$ of rank $0$. Through this identification, the domain of $g_k$ and $f_k$ may also be equal; in that case we may conclude that $g_1 \circ \dots \circ g_k \neq f_k$.

The relevance of Theorem~\ref{thm:engelfrieteries} to the question of pebble minimization comes from:
\begin{theorem}[Engelfriet \& Maneth~{\cite[item (2) of the abstract]{EngelfrietPebbleMacro}}]
  \label{thm:pebble-to-macro}
  Any tree-to-tree function computed by some $k$-pebble\footnote{As explained in \Cref{ftn:off-by-one} of the introduction, the indexing convention for the pebble transducer hierarchy in~\cite{EngelfrietPebbleMacro} is off by one compared to ours.} tree transducer can also be expressed as a $k$-fold composition of macro tree transducers. (But the converse is false.)
\end{theorem}

A $k$-pebble string transducer is none other than a $k$-pebble tree transducer working on the encodings of strings as unary trees described above. Thus it follows directly (without having to recall the definition of $\Pebble_k$ from \S\ref{sec:pebble-def}) that $f_k\notin\Pebble_k$, and even that $\Image(f_k)\neq\Image(h)$ for any $h\in\Pebble_k$.
On the other hand, any two-dimensional FO interpretation is polyregular
(Theorem~\ref{thm:interp-polyreg}) with quadratic growth (an immediate
consequence of the definition is that the output length is at most the square of
the input length). So the $f_k$ are indeed counterexamples to pebble
minimization.

\begin{remark}\label{rem:fo-useless}
  We have defined the $f_k$ by first-order interpretations because they provide
  a convenient notion of \enquote{two-dimensional origin semantics}, used in our
  inductive construction. But it would have been possible to show by a simple
  ad-hoc argument that each $f_k$ is in some $\Pebble_\ell$ (as defined in
  \Cref{sec:pebble-def}). Therefore, we do not depend in an essential way on the
  difficult translation from FO interpretations to pebble transducers
  (Theorem~\ref{thm:interp-polyreg}) to refute pebble minimization.
\end{remark}

 We recall the required properties of macro tree transducers in
 \Cref{sec:comp-mtt}, then we prove the main Theorem~\ref{thm:engelfrieteries} in \Cref{sec:engelfrieteries-proof}.

\subsection{Compositions of macro tree transducers (MTTs)}
\label{sec:comp-mtt}

We do not formally define MTTs here, but only recall useful facts for our purposes. In this paper, we only consider \emph{total deterministic} MTTs. Let $\MTT^k$ be the class of tree-to-tree functions computed by some composition of $k\geq1$ macro tree transducers.
\begin{remark}
  $\MTT^k$ can be characterized equivalently as the class of functions computed
  by \mbox{$k$-iterated} pushdown transducers~\cite{EngelfrietPushdownMacro}, or by
  \enquote{level-$k$} tree transducers~\cite{EngelfrietHighLevel}. For any $k$,
  the functions in $\MTT^k$ with linear growth are precisely the regular tree
  functions~\cite{EngelfrietIM21}\footnote{The paper~\cite{EngelfrietIM21} talks
    about compositions of tree-walking transducers (TWT), but $\MTT^1$ is
    included in the class of functions obtained by composing 3
    TWTs~\cite[Lemma~37]{EngelfrietPebbleMacro}.} -- this, combined with
  Theorem~\ref{thm:pebble-to-macro}, proves pebble minimization for
  polyregular functions of linear growth,\footnote{This special case is also a
    consequence of dimension minimization for MSO interpretations (Remark~\ref{rem:mso-interp}), thanks to the equivalence between two-way (i.e.\ 1-pebble) transducers and MSO transductions~\cite{EngelfrietHoogeboom}.} showing that our counterexamples with quadratic growth are in some sense minimal.
\end{remark}

For a class of tree-to-tree functions $\cC$, we also write $\SO(\cC)$ for the subclass of functions that output strings (via the identification described at the beginning of \S\ref{sec:engelfrieteries}). The literature on tree transducers often refers to the classes $\SO(\MTT^k)$ by equivalent descriptions involving a \emph{yield} operation that maps trees to strings:
$\yield(t)$ is the word formed by listing the labels of the leaves of $t$ in infix order.
\begin{claim}\label{clm:clarification-mtt}
  $\SO(\MTT^1) = \{\yield\circ f \mid f\ \text{is computed by some \emph{top-down} tree transducer}\}$, and for any $k\in\bbN$, we also have
  $\SO(\MTT^{k+2}) = \{\yield\circ f \mid f \in \MTT^{k+1}\}$.
\end{claim}
\begin{proof}
  The claim about $\SO(\MTT^1)$ is well-known. For instance it is stated by Maneth~\cite[end of~\S5]{Maneth15} as follows:
  \enquote{Note that macro tree transducers with monadic output alphabet are essentially the same as top-down tree-to-string transducers} -- where, as is the tradition in tree transducer papers, \enquote{top-down tree-to-string} means $\yield\circ(\text{top-down tree-to-tree})$, \emph{not} $\SO(\text{top-down tree-to-tree})$!

  We can then deduce the second claim:
  \begin{align*}
    \SO(\MTT^{k+2}) &= \SO(\MTT^1)\circ\MTT^1\circ\MTT^{k}\\
    &= \yield \circ \underbrace{(\text{top-down tree transducers}) \circ \MTT}_{=\,\MTT,\ \text{cf.~\cite[Lemma~5]{OutputMacro}}} \circ\; \MTT^{k}
  \end{align*}
  (beware: in~\cite{OutputMacro}, the notation $\circ$ is flipped compared to ours).
\end{proof}

This allows us to rephrase a key \enquote{bridge theorem} by Engelfriet and Maneth~\cite{OutputMacro} in a form that suits us better. For a class of string-valued functions $\cC'$, let $\Image(\cC') = \{\Image(f) \mid f \in \cC'\}$.
\begin{theorem}[{\cite[Theorem~18]{OutputMacro} + Claim~\ref{clm:clarification-mtt}}]
  \label{thm:bridge}
  Let $k \geq 1$ and $L,L'$ be string languages. If $L' \in
  \Image(\SO(\MTT^{k+1}))$ and $L'$ is \emph{d-complete} (see below) for $L$, then $L\in\Image(\SO(\MTT^k))$.
\end{theorem}
\begin{definition}\label{def:d-complete}
  Let $L \subseteq \Sigma^*$ and $L' \subseteq (\Sigma \cup \Delta)^*$ with $\Sigma \cap \Delta = \varnothing$. We say that $L'$ is:
    \begin{itemize}
    \item \emph{$\delta$-complete for $L$}~\cite[\S5]{OutputMacro} when for every $u \in L$ there exist $w_0, \dots, w_n \in \Delta^*$ such that
    \begin{itemize}
      \item all the $w_i$ for $i\in\{1,\dots,n-1\}$ are \emph{pairwise distinct} words;
      \item $n = |u|$ and $w_0 u[1] w_1 \ldots u[n] w_n \in L'$;
    \end{itemize}
    \item \emph{d-complete for $L$}~\cite[\S4]{PebbleString} when it is $\delta$-complete for $L$ and \enquote{conversely}, by erasing the letters from $\Delta$ in the words in $L'$ one gets exactly $L$,
    i.e.\ $\varphi(L') = L$ where $\varphi$ is the monoid morphism mapping each letter in $\Sigma$ to itself and $\Delta$ to $\varepsilon$.
  \end{itemize}
\end{definition}

\subsection{Proof of Theorem~\ref{thm:engelfrieteries}}
\label{sec:engelfrieteries-proof}

Our strategy is a minor adaptation of Engelfriet and Maneth's proof~\cite[\S4]{PebbleString} that the hierarchy of output languages of $k$-pebble string transducers is strict.
(See also~\cite[Theorem~41]{EngelfrietPebbleMacro} for a similar result on tree languages.)
Writing $\Image(\cI) = \Image(f)$ when the first-order interpretation $\cI$ defines the string function $f$, we show that:
\begin{lemma}\label{lem:engelfrieteries}
  From any two-dimensional first-order interpretation $\cI$, one can build another 2D FO interpretation $\Psi(\cI)$ such that $\Image(\Psi(\cI))$ is d-complete for $\Image(\cI)$.
\end{lemma}
Let us explain how this lemma leads to Theorem~\ref{thm:engelfrieteries}. For $k\geq1$, let
\[ \cI_k = \Psi^{k-1}(\text{the FO interpretation of
    Example~\ref{exa:interp0} that defines}\ a^n \mapsto (a^{n-1}b)^{n-1}) \]

We have\footnote{This implies $(a^n \mapsto (a^{n-1}b)^{n-1})\notin\MTT^1$, a
  fact that was later reproved in~\cite[Theorem~8.1(a)]{NguyenNP21}.}
$\Image(\cI_1)\notin\Image(\SO(\MTT^1))$ by immediate application
of~\cite[Theorem~3.16]{Engelfriet82}.\footnote{This old result of Engelfriet says that some superclass of $\Image(\SO(\MTT^1))$ does not contain any language of the form $\{(a^n b)^{f(n)} \mid n \in X\}$ where $X \subseteq \bbN\setminus\{0\}$ is infinite and $f \colon X \to \bbN\setminus\{0\}$ is injective.} By induction on $k$, we then get $\Image(\cI_k)\notin\Image(\SO(\MTT^k))$ for all $k$ -- which is precisely what we want -- using Lemma~\ref{lem:engelfrieteries} together with the contrapositive of Theorem~\ref{thm:bridge} for the induction step.

Our only remaining task is to prove Lemma~\ref{lem:engelfrieteries}.
\begin{proof}
  Let $\cI$ be a 2D FO interpretation from $\Gamma^*$ to $\Sigma^*$. 
  Choose $\clubsuit\notin\Gamma$ and $\square,\lozenge\notin\Sigma$ with $\square\neq\lozenge$.
  We define a new function $f'\colon (\Gamma\cup\{\clubsuit\})^* \to (\Sigma\cup\{\square,\lozenge\})^*$, which coincides with $f$ on $\Gamma^*$, in the following way. First, let us give an example: 
  for the interpretation of Example~\ref{exa:interp0} defining $a^n \mapsto (a^{n-1}b)^{n-1}$, the new function would map
  (the colors below serve to highlight blocks of matching sizes)
  \[ \clubsuit \dots \clubsuit a{\color{red}\clubsuit\clubsuit\clubsuit} a\clubsuit a{\color{blue}\clubsuit\clubsuit} \quad\text{to}\quad a{\color{red}\square\square\square}{\color{red}\lozenge\lozenge\lozenge} a{\color{red}\square\square\square}\lozenge b{\color{red}\square\square\square}{\color{blue}\lozenge\lozenge}a\square{\color{red}\lozenge\lozenge\lozenge} a\square\lozenge b\square{\color{blue}\lozenge\lozenge}. \]
  In general, for any input $u = u_1 \dots u_n \in \Gamma^*$, let $f(u) = v_1 \dots v_m$ (where the $u_i$ and $v_i$ are letters).
  Reusing the notation $O^\cI_u$ from Definition~\ref{def:fo-interp}, let
  $O^\cI_u = \{ \langle i_1, j_1\rangle \leq \dots \leq \langle i_m, j_m\rangle \}$ -- morally, $\langle i_p, j_p \rangle$ is the \enquote{origin} (as in \Cref{sec:reg-origins}) of the output letter $v_p$. Then, for any $p \in \bbN^{\{0,\dots,n\}}$, we take
  \[ f'\left(\clubsuit^{p[0]} u_1 \clubsuit^{p[1]} \dots u_n \clubsuit^{p[n]}\right) = v_1 \square^{p[i_1]} \lozenge^{p[j_1]} \dots v_m \square^{p[i_m]} \lozenge^{p[j_m]} \]
  -- note that the prefix in $\clubsuit^*$ is ignored.
  (Our use of the blocks of $\clubsuit$s is inspired by the coding of atom-oblivious functions in the Deatomization Theorem of~\cite{PolyregGrowth}, cf.\ Remark~\ref{rem:deatomization}.)

  We derive from $\cI$ a new 2D FO interpretation $\psi(\cI)$ that specifies $f'$.
  For example, we can illustrate $\Psi(\text{Example~\ref{exa:interp0}})$ by the figure below:

  \begin{center}
    \includegraphics{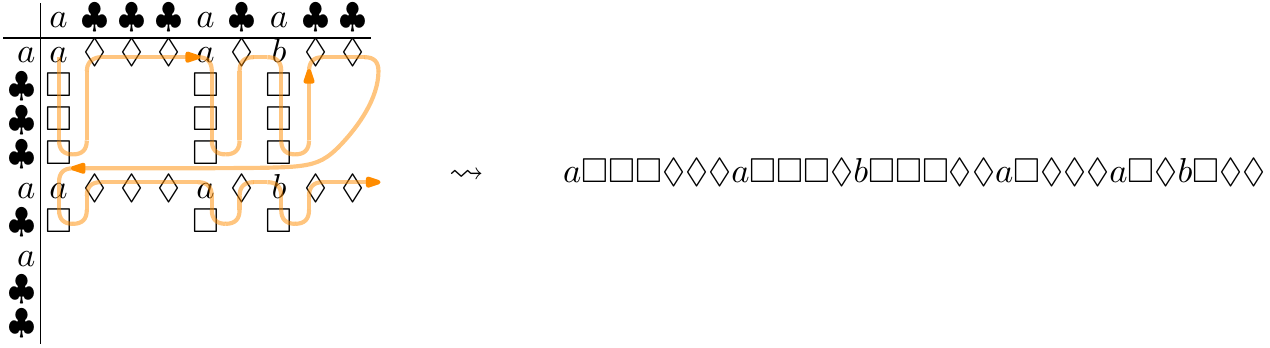}
  \end{center}

  Now let us give the general recipe for $\Psi(\cI$).
  Given a formula $F$ over the relational signature for $\Gamma^*$, let $F^\relat$ be the relativized formula over $(\Gamma\cup\{\clubsuit\})$ where all
    quantifiers $\forall z.(\dots)$ and $\exists z.(\dots)$ are replaced by $\forall z.\; \lnot\clubsuit(z) \Rightarrow (\dots)$ and $\exists z.\;
    \lnot\clubsuit(z) \land (\dots)$ respectively. Let $P(x,y)$ be a formula stating that
    the restriction of the input to $[x,y]$ is in $\Gamma \clubsuit^*$:
  \[ P(x,y) = x \leq y \land \lnot\clubsuit(x) \land \forall z.\; (\lnot(z \leq x) \land z \leq y) \Rightarrow \clubsuit(z) \]
  We take the following definition for the interpretation $\Psi(\cI)$: 
  \begin{itemize}
    \item $\Psi(\cI)_c(x_1,x_2) = \lnot\clubsuit(x_1)\land\lnot\clubsuit(x_2)\land\cI_c^\relat(x_1,x_2)$ for $c\in\Sigma$
    \item $\Psi(\cI)_\square(x_1,x_2) = \clubsuit(x_1)\land\lnot\clubsuit(x_2)\land \exists \widehat{x_1}.\; P(\widehat{x_1},x_1) \land \displaystyle\bigvee_{c\in\Sigma} \cI_c^\relat(\widehat{x_1},x_2)$
    \item $\Psi(\cI)_\lozenge(x_1,x_2) = \lnot\clubsuit(x_1)\land\clubsuit(x_2)\land \exists \widehat{x_2}.\; P(\widehat{x_2},x_2) \land \displaystyle\bigvee_{c\in\Sigma} \cI_c^\relat(x_1,\widehat{x_2})$
    \item $\Psi(\cI)_\leq(x_1,x_2,y_1,y_2) = \exists \widehat{x}_1,\widehat{x}_2,\widehat{y}_1,\widehat{y}_2.\; \displaystyle\bigwedge_{i=1,2} P(\widehat{x}_i,x_i) \land P(\widehat{y}_i,y_i)$ and
    \begin{itemize}
    \item either $(\widehat{x}_1,\widehat{x}_2) \neq (\widehat{y}_1,\widehat{y}_2)$ and $\cI_\leq^\relat(\widehat{x}_1,\widehat{x}_2,\widehat{y}_1,\widehat{y}_2)$
      \item or $(\widehat{x}_1,\widehat{x}_2) = (\widehat{y}_1,\widehat{y}_2)$ and $(x_2,x_1) \leq (y_2,y_1)$ lexicographically.
    \end{itemize}
  \end{itemize}
  Finally, we claim that $\Image(f')$ is d-complete for $\Image(f)$.
  The part concerning the erasing morphism is immediate.
  For $\delta$-completeness (Definition~\ref{def:d-complete}), let $v\in\Image(f)$, i.e.\ $v = f(u) = f(u_1 \dots u_n)$ for some $u\in\Gamma^*$; then $f'(u_1\clubsuit u_2 \clubsuit\clubsuit \dots u_n\clubsuit^n) \in \Image(f')$ is of the required form.
\end{proof}

\begin{remark}
  The construction in our proof of Lemma~\ref{lem:engelfrieteries} preserves the class of interpretations that verify the property described in the proof of Theorem~\ref{thm:polyblind-refutation}.
\end{remark}

\bibliographystyle{alphaurl}
\bibliography{bib}


\end{document}